\documentclass{article}

\usepackage{arxiv}

\usepackage[utf8]{inputenc} 

\usepackage{makeidx}  
\usepackage[english]{babel}
\usepackage{graphics}
\usepackage{graphicx}
\usepackage{amssymb}
\usepackage{amsthm}
\usepackage{verbatim}
\usepackage{color}
\usepackage{tcolorbox}

\newtheorem {theorem}{Theorem}
\newtheorem {remark}{Remark}

\newtheorem {corollary}{Corollary}
\newtheorem{proposition}{Proposition}

\author{Eric Goles	\thanks{ANID via  FONDECYT 1200006 + STIC- AmSud CoDANet project 19-STIC-03} \\
	Facultad de Ingenier\'ia y Ciencias, Universidad Adolfo Ib\'a\~nez, Santiago, Chile \\
	\texttt{eric.chacc@uai.cl} \And
	Marco Montalva-Medel \thanks{STIC- AmSud CoDANet project 19-STIC-03.} \\
	Facultad de Ingenier\'ia y Ciencias, Universidad Adolfo Ib\'a\~nez, Santiago, Chile \\
	\texttt{marco.montalva@uai.cl} \And
	Pedro Montealegre \thanks{ANID via  PAI + Convocatoria Nacional Subvenci\'on a la Incorporaci\'on en la Academia A\~no 2017 + PAI77170068, FONDECYT 11190482, FONDECYT 1200006,  
			STIC- AmSud CoDANet project 19-STIC-03.}. \\
			Facultad de Ingenier\'ia y Ciencias, Universidad Adolfo Ib\'a\~nez, Santiago, Chile \\
			\texttt{p.montealegre@uai.cl} \And
	Mart\'in R\'ios-Wilson. \thanks{ANID via PFCHA/DOCTORADO NACIONAL/2018 – 21180910 + PIA AFB 170001, French ANR project FANs ANR-18-CE40-0002 and ECOS project C19E02} \\
	Departamento de Ingeniería Matemática, FCFM, Universidad de Chile, Santiago, Chile. \\
	Aix Marseille Univ, Université de Toulon, CNRS, LIS, Marseille, France. \\
	\texttt{mrios@dim.uchile.cl}
}

\title{On the complexity of the generalized Q2R automaton}
\date{\vspace{-5ex}}
\begin{document}
\maketitle		
\vspace{-5ex}
		\begin{abstract}
			We study the dynamic and complexity of the generalized Q2R automaton. We show the existence of non-polynomial cycles as well as its capability to simulate with the synchronous update the classical version of the automaton updated under a block sequential update scheme. Furthermore, we show that the decision problem consisting in determine if a given node in the network changes its state is \textbf{P}-Hard. 
			
		\end{abstract}
\keywords{Q2R networks, computational complexity, limit cycles, \textbf{P}-Complete}	
	
	
	
	\section{Introduction}
	
	
	
	In this paper we study the reversible cellular automata Q2R rule, introduced by G. Vichniac in the mid-80's \cite{vichniac1984simulating} as a representation of the two-dimensional Ising model for ferromagnetism \cite{lenz1920beitrag,ising1925beitrag,brush1967history}.   The Q2R rule is defined as follows. Consider a two-dimensional finite grid of even size, with periodic boundaries and von-Neumann neighborhood. Each cell has one over two possible states, namely $-1$ and $1$, which  evolve according to the following dynamic. The cells are divided in two blocks, consisting in the white and black cells of a checkerboard. On each time-step, all cells in black squares are iterated, and after that, is the turn of the cells in white squares. An iteration means to update the current state according to the following rule: if the number of neighbors in state $1$ equal the number of neighbors in state~$-1$, the cell switches its current state to the opposite one. Otherwise, the cell remains in its current state. Formally, the local transition function $f$ of Q2R~is:
	
	$$ f_i(x)= \left\{  \begin{array}{cl} 1-x_i , & \textrm{if } \displaystyle\sum_{j\in N(i)} x_j = 0\\ x_i ,& \textrm{otherwise} \end{array}\right. , $$
	where $x \in \{-1, 1\}^n$ is a configuration of $n$ cells, $i$ is a cell and $N(i)$ represent the set of cells 
	in the von-Neumann neighborhood of $i$.

	Several studies have carried out concerning the Q2R \cite{vichniac1984simulating,pomeau1984invariant,herrmann1986fast,hermann1987periods,takesue1991relaxation,cordero1992q2r+,goles2011irreversibility,urbina2016master,montalva2020phase}. Remarkably, the Q2R dynamic preserves an Ising-like energy \cite{pomeau1984invariant}, appealing the analogy with the continuous dynamics of Hamiltonian systems.  Indeed, let us consider the following energy operator: 
	
	$$ \textit{E}(x) = -\frac{1}{2}\sum_{ij }w_{ij}x_ix_j$$ 
	where $x$ is a grid configuration of states, and $w_{ij}$ 
	is equal to $1$ if cell $i$ is adjacent to cell $j$, and 0 otherwise.
	Let us call $x'$ the configuration obtained after the iteration of one of the two parts (cells in black squares or cells in white squares), then
	
	$$\Delta\textit{E} = E(x') - E(x) = -\frac{1}{2} \sum_i (x'_i-x_i)\sum_{j\in N(i)}x_j = 0,$$
	which implies that the energy is preserved during the dynamics of Q2R. Besides, since the dynamics of Q2R is reversible, every configuration has a unique predecessor. In particular, the dynamics does not exhibit a transient state. Thus, each configuration is in some attractor, i.e., the configurations are fixed points or they belong to a limit cycle of some period.
	
	Recently, in an attempt to establish general mathematical properties, a full characterization and combinatorial results of the attractors associated to the Q2R model were proposed in \cite{montalva2020phase}. In tune with this mathematical approach, in this paper we tackle an analytical study of the dynamics and complexity of Q2R. In this case, we consider a \emph{generalized version of Q2R}, extending its definition to a topology more general than the two-dimensional grid. 
	
	An \emph{automata network} is a natural generalization of a finite (or periodic) cellular automata, 
	where the topology representing the interactions on the cells is generalized into an arbitrary graph. We extend the Q2R in this direction. Nevertheless, to preserve the reversibility of the system and a coherent definition of the local rule, we do not extend this rule into an arbitrary topology, but into a graph that is bipartite and where each node has an even number of neighbors. 
	
	Formally, we consider the family of graphs $\mathcal{G}$, such that each graph $G\in \mathcal{G}$ is a bipartite graph with partitions $A$ and $B$, and each vertex has even degree. Over each of these graphs $G$, we define the generalized Q2R as follows. Each node of $G$ is assigned a state in $\{-1,1\}$, which evolves according to the following rule. On each time-step, all the nodes in $A$ are synchronously iterated, and after that, all the nodes in $B$ are synchronously iterated. In other words, the sets $A$ and $B$ play the role of black and white squares in the chessboard. The iteration of a node consists in the application of the local rule, that switches the current state of the node when the number of neighbors in state $-1$ is the same than the number of neighbors in state $+1$.  
	
	Interestingly, in a one dimensional grid, the Q2R rule corresponds to the elementary cellular automata 150 on finite configurations, updated following the bipartite partition induced by the cells in even and odd coordinates.
	Numerically, it has been reported that the classical Q2R rule can exhibit a huge number of limit cycles with hypothetic exponentially long periods \cite{hermann1987periods}. 
	
	
	The computational complexity of an automata network can be defined as the amount of resources, like time or space, needed to make predictions over it. More formally, the \emph{prediction problem} consists in given an automata network, an initial condition, an objective node in the network and a fixed time $t$, to predict if the objective node will change its state after $t$ time steps. From a classical computational complexity standpoint, the latter problem can be studied as a decision problem. Since decision problems are classified in \emph{complexity classes}, with the objective of grouping problems of similar difficulty, a natural question is in which class we can classify the prediction problem for Q2R networks.  In this sense, the numeric behavior of Q2R represent a great challenge for the simulation of the dynamics of Q2R.  In fact, since the dynamics of Q2R can exhibit long cycles, one might have to simulate a huge number of time-steps before deciding. Interestingly,  up to our knowledge, the computational complexity of the prediction problem for Q2R is unknown.

	\subsection{Our results}
	
	During this paper, we focus in studying the dynamical behavior and the complexity of the generalized Q2R model. First, we give examples 
	of graph topologies in which the generalized Q2R automata network can exhibit super-polynomial limit cycles. Our construction consists on a specific graph, called cycle-graph, which given a particular initial configuration reach a limit cycle that is proportional to the number of nodes in the cycle-graph. Then, we show how to connect two different or more cycle-graphs in such a way that the dynamics of each cycle-graph is preserved. We obtain this way a new graph on which we define a configuration reaching a limit cycle of a length that is proportional to the least common multiple of the sizes of the cycle-graphs composing it. The result then is obtained by constructing a graph composed of cycle-graphs of different prime numbers.
	
	We emphasize that, even if numerical evidence suggests that Q2R can exhibit super-polynomial (or even exponential) limit cycles, our result is the first explicit example of a (generalized) Q2R dynamic where a super-polynomial lower-bound on the limit cycles can be analytically proven. 
	
	Later, we tackle the computational complexity of the rule, showing that the prediction problem is \textbf{P}-hard. This roughly implies that the best (only) strategy for knowing the future state of a node is to simply simulate the dynamics for a given number of time-steps. In other words, unless some complexity classes collapse, there are no structural or dynamical properties of the network that could be algorithmically exploited in order to solve the prediction problem faster than naive simulation. Our result is obtained by a reduction of the \emph{circuit value problem}.
	
	Finally, we show that the dynamics of the (generalized) Q2R networks updated by a bipartite partition can be simulated in another Q2R network with synchronous update. Such a network, that we will call it {\it Parallel Simulator (PS)},  is composed by four subnetworks, each one having a specific function. The latter network simulates one step of the dynamics of the original Q2R network every two time steps.
	
	
	The paper is organized as follows: Section \ref{sec:prelim} summarizes the basic mathematical
	concepts and definitions that we will use along the manuscript. 
	
	In Section \ref{sec:q2r-pspace-c} we exhibit Q2R networks having non-polynomial limit cycles. We accomplished this task by exhibiting for each prime number $p \in \mathbb{N}$ some Q2R network having attractors of period $\mathcal{O}(n)$.
	
	In Sections \ref{sec:pred} and \ref{sec:gadgets} we study the  complexity of the prediction problem. In particular, we show in Section \ref{sec:pred} that the latter problem is \textbf{P}-Hard by constructing a log-space reduction from a version of the classical circuit value problem (\textbf{CVP}) in which all the gates are only monotone gates (OR and AND gates).
	
	Finally, in Section \ref{sec:ps} we show that we can simulate an arbitrary  Q2R network (which is defined by bipartite partition) by some Q2R rule iterated in parallel (every vertex at the same time). This construction uses a polynomial amount of space in the size of the original network. Finally, we give some conclusions and future developments.  
	
	
	\section{Preliminaries}\label{sec:prelim}
	
	An {\it Automata Network} is a triple $\mathcal{A} = (G,Q,(f_i: i\in V))$, where $G = (V,E) $ be an undirected finite bipartite graph,  $V = \{1, ..., n\}$ is the set of  vertices, $E$ the set of edges and $Q = \{-1,1\}$ is the set of states. The state 1 means that the vertex is {\it active}, while state -1 represents {\it passive} vertices;  $f_i: Q^{\vert V \vert} \rightarrow Q$ is the transition function associated to the vertex $i$. The set $Q^{\vert V \vert}$ is called the set of configurations, and the automaton's global transition function $F: Q^{\vert V \vert} \rightarrow Q^{\vert V \vert}$, is constructed from the local functions $ (G,Q,(f_i: i\in V))$ such that $(F(x))_i = f_i(x)$.  Let $N(v)$ be the neighborhood of $v$, i.e. the set of vertices $\{u \mid uv \in E\}$. Suppose that
	$\{A,B\}$ is a partition of V such that for $X\in\{A,B\}$ and for any $i,j \in X$ ,  $N(i)\cap N(j)=\emptyset$ (such partition exists because G is bipartite). Further, consider that vertices have even degree, i.e., $d(v)\equiv |N(v)|$ is even.  We define a Q2R network as the tuple $\mathcal{Q} = (G=(V=(B,W), E), F)$ where $G$ is a bipartite graph with bipartition $V=(B,W)$ and $F: \{0,1\}^{V} \mapsto \{0,1\}^{V}$ is such that $F(x)_{i} = f_{i}(x|_{N(i)})$ for all $x \in \{0,1\}^{V}$ where $ f_i(x)= \left\{  \begin{array}{cl} 1-x_i , & \textrm{if } \displaystyle\sum_{j\in N(i)} x_j = 0\\ x_i ,& \textrm{otherwise} \end{array}\right. , $ and $x|_{N(i)} \in \{0,1\}^{V}$ is such that $(x|_{N(i)})_{v} = x_{v}$ for all $v \in N(i).$
	
	\subsection{Elements of computational complexity}
	The computational complexity of a decision problem is defined as the amount of resources (time or space) needed to give an answer. The classical complexity theory consider two fundamental classes: {\bf P}, the class of problems solvable in polynomial time on a serial computer; and {\bf NC}, the class of problems solvable in poly-logarithmic time with a polynomial amount of processors in a parallel architecture, for instance a PRAM machine \cite{jaja}. It is easy to prove that ${\bf NC} \subset {\bf P}$ \cite{jaja}. Informally {\bf NC} is known as the class of problems which have a fast parallel algorithm \cite{limitsofpara}. It is a well known conjecture that \textbf{NC} $\neq$ \textbf{P}, and if so, there exist ``inherently sequential'' problems, this is, that belong to \textbf{P} and do not belong to \textbf{NC}. The most likely to be inherently sequential are {\bf P}-Hard problems, to which any other problem in \textbf{P} can be reduced by an \textbf{NC}-reduction or a logarithmic space reduction. If a {\bf P}-Complete problem (i.e. a {\bf P}-Hard problem contained in {\bf P}) has a fast parallel algorithm, then \textbf{P}=\textbf{NC} \cite{limitsofpara,CMoore}.
	
	One of such problems is the {\it Circuit Value Problem} ({\bf CVP}), which consists in predicting the truth value of the output of a Boolean circuit. This problem is \textbf{P}-complete since any deterministic Turing machine computation of length $k$ can be converted into a Boolean circuit of depth $k$;  a complete analysis of this reduction can be found in \cite{limitsofpara}. 
	Given a circuit, we define the {\it layer} of a gate $v$, denoted $layer(v)$, as follows: it is zero for the input gates and the 
	length of the longest path from an input to $v$. A circuit is {\it synchronous} if all inputs to a gate $v$ come from gates at precedent layer. Furthermore, we require that all output vertices be on the same layer, namely the highest \cite{limitsofpara}. 
	A circuit is monotone if there are no negations gates (only AND and OR gates); it is alternating when the the gates alternate between OR and AND gates layer by layer, and the inputs are connected only to OR gates, and the outputs being OR gates.
	
	The {\bf CVP} remains \textbf{P}-complete when the circuit is restricted to be synchronous, monotone, alternating  and all vertices have in degree (fan in) and out degree (fan out) exactly two, with the obvious exceptions of the input with in degree zero, and the outputs with out degree zero \cite{limitsofpara}. We call {\bf AS2MCVP} this restriction of the {\bf CVP}. 
	
	\subsubsection{Prediction problem}
	
	We start by providing a formal definition for the prediction problem in Q2R networks. \\
	
	\begin{tcolorbox}
		\textbf{PRED}
		
		\textbf{INPUT: } A Q2R network $\mathcal{Q} =  (V=(B,W),E,F)$, an initial condition $x \in \{0,1\}^{n}$,  a time step $t \in \mathbb{N}$ and an objective node $v \in V.$
		
		\textbf{QUESTION:} $F^{t}(x)_{v} \not = x_{v}?$
	\end{tcolorbox}
	
	
	
	
	
	Observe that a possible solution for \textbf{PRED} is just to simulate the system in order to compute $F^{t}(x)_{v}$ and see if it has changed its state.  Thus, a natural question is if one can do better than simple simulation. As we will see in Section  \ref{sec:pred}, the problem $\textbf{PRED}$ is $\textbf{P}$-hard, which suggest, together with the fact that Q2R networks may have limit cycles of non-polynomial period, that there is unlikely that a better solution than direct simulation exists.  
	\section{Non polynomial limit cycles for Q2R}\label{sec:q2r-pspace-c} 
	
	In this section we construct a family of simple one dimensional networks (rings)  which admits a prime cycle depending of the ring size. Further by connecting those networks we determine periods which are roughly the product of elementary rings. The following Proposition and consequent Corollary states that Q2R networks can exhibit large limit cycles.
	
	\begin{proposition}\label{lem:r150-lp}
		If $p\geq 2$ is a prime number, then there exists a connected Q2R network of $2p$ vertices having, at least, one limit cycle of length $p$.
	\end{proposition}
	\begin{proof}
		Let $p\geq 2$ be a prime number and considerer the Q2R network of Figure \ref{fig:r150-lp} defined by the ECA rule 150 with $2p$ vertices, where the odd vertices are updated first and, in second place, the even ones. 
		\begin{figure}[h] 
			\centering
			\includegraphics[scale=0.9]{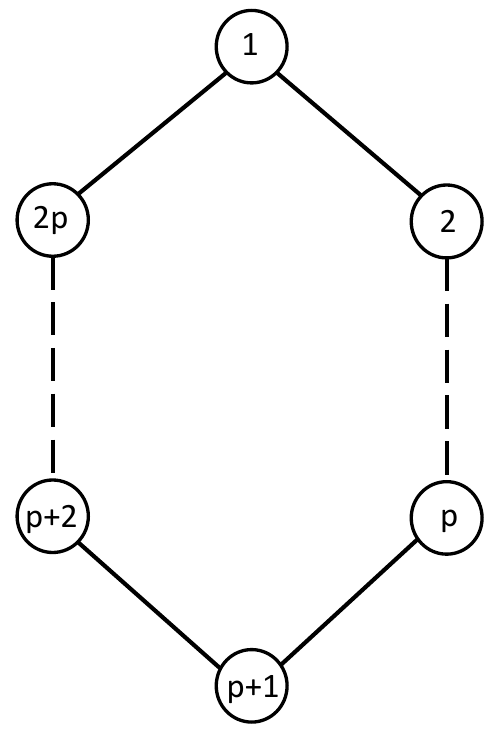} 
			\caption{Q2R network of $2p$ vertices ($p\geq 2$ a prime number).}
			\label{fig:r150-lp}
		\end{figure}
		
		Let $x=1\overrightarrow{0}$ be the configuration of length $2p$, where all its ($2p-1$) states are in 0, except for the first one, whose state is 1. If $p=2$ then we obtain the limit cycle of length two: $1000\longleftrightarrow 1101$. If $p\geq 3$, $p$ prime, notice the key fact: except for vertices 1 and $p+1$, the others ones always are updated in pairs, i.e., 
		$$x_k = x_{2p-(k-2)}, \,\forall k\in\{2,...,p\}.$$
		Thus, at each time step, the active vertices are propagated in pairs (or blocks of pairs) until reach the configuration in which only the vertex $p + 1$ is passive and, after that, the passive vertices are propagated until reach the initial configuration $x=1\overrightarrow{0}$. Specifically, at $t=1$ only vertex 1 is active (the initial configuration $x=1\overrightarrow{0}$), at $t=2$ also are activated the vertices $2$ and $2p$ (the configuration $x=11\overrightarrow{0}1$ of length $2p$) but, from $t=3$ two pairs of vertices are activated; the activation of vertices 3 and 4 (consequently, vertices $2p-1$ and $2p-2$) for $t=3$,  the activation of vertices 5 and 6 (consequently, vertices $2p-3$ and $2p-4$) for $t=4$, and so on, until the activation of vertices $p-2$ and $p-1$ (consequently, vertices $p+2$ and $p+3$) for $t=2+\frac{p-3}{2}$. Next, at $t=3+\frac{p-3}{2}$ we have reached the configuration in which only the vertex $p + 1$ is passive. This process is repeated by deactivating two pairs of vertices until reach the last configuration of the limit cycle at $t=\left(3+\frac{p-3}{2}\right)+\frac{p-3}{2}$ (see Figure \ref{fig:q2r-evol}).
		
		\begin{figure}[h] 
			\centering
			\includegraphics[scale=0.55]{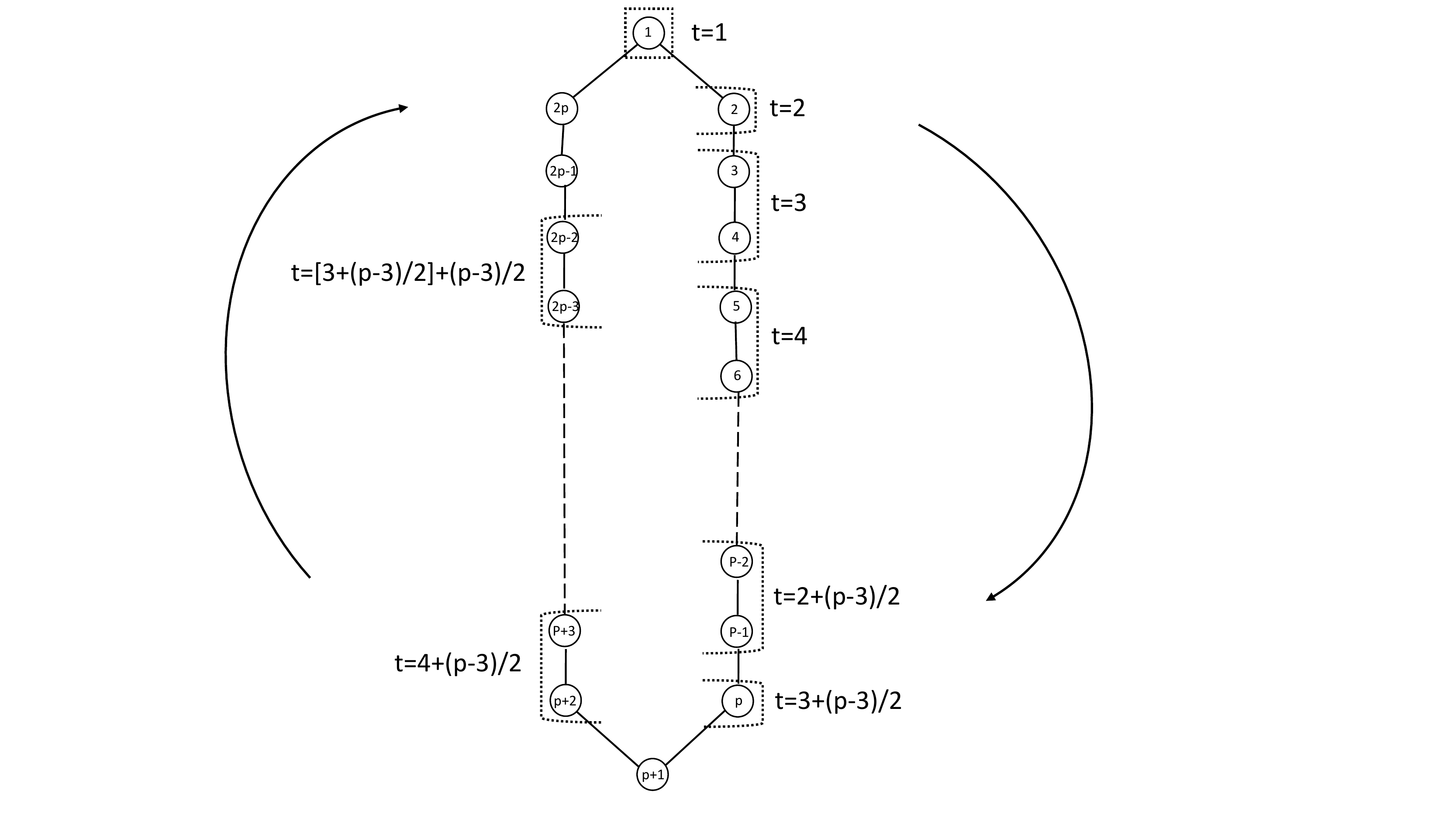} 
			\caption{Evolution of the initial configuration $x=1\overrightarrow{0}$ (of length $2p$) in the Q2R network of Figure \ref{fig:r150-lp}. The dashed brackets indicate the vertices involved in the corresponding time step, according to what is detailed in the proof of Proposition \ref{lem:r150-lp}. }
			\label{fig:q2r-evol}
		\end{figure}
		
		Therefore, we have a limit cycle of length $p$.
	\end{proof}
	
	\begin{remark}\label{rem:q2r-2p-dynamic}
		Notice the following facts for the Q2R networks of $2p$ vertices described in the proof of Proposition \ref{lem:r150-lp}:
		\begin{itemize}
			\item[(i)] In the configuration $x=1\overrightarrow{0}$, the state of the first vertice always remains fixed to 1 (at every time step, it never changes to 0) because of the key fact: $x_2 = x_{2p}$. Analogously, the state of the vertice $p+1$ always remains fixed to 0 (at every time step, it never changes to 1) because of the key fact: $x_p = x_{p+2}$.
			\item[(ii)] They only have four fixed points which are: \overrightarrow{0}, \overrightarrow{1}, \overrightarrow{01} and \overrightarrow{10}.
		\end{itemize}
	\end{remark}
	
	\begin{corollary} \label{coro:q2r-exp-lc} There exists a connected Q2R network of $2\cdot (p_1+p_2+\cdots +p_k)$ vertices having, at least, one limit cycle of length $p_1\cdot p_2\cdots p_k$, where $p_1$,..,$p_k$ are $k$ different prime numbers. In particular, there exists a Q2R network which exhibits attractors of period $T \geq e^{\Omega(\sqrt{\vert V(G) \vert \log(\vert V(G) \vert) })}$
	\end{corollary}
	\begin{proof}
		Consider the connected network $G$ of Figure \ref{fig:q2r-superpol}.
		\begin{figure}[h] 
			\centering
			\includegraphics[scale=0.5]{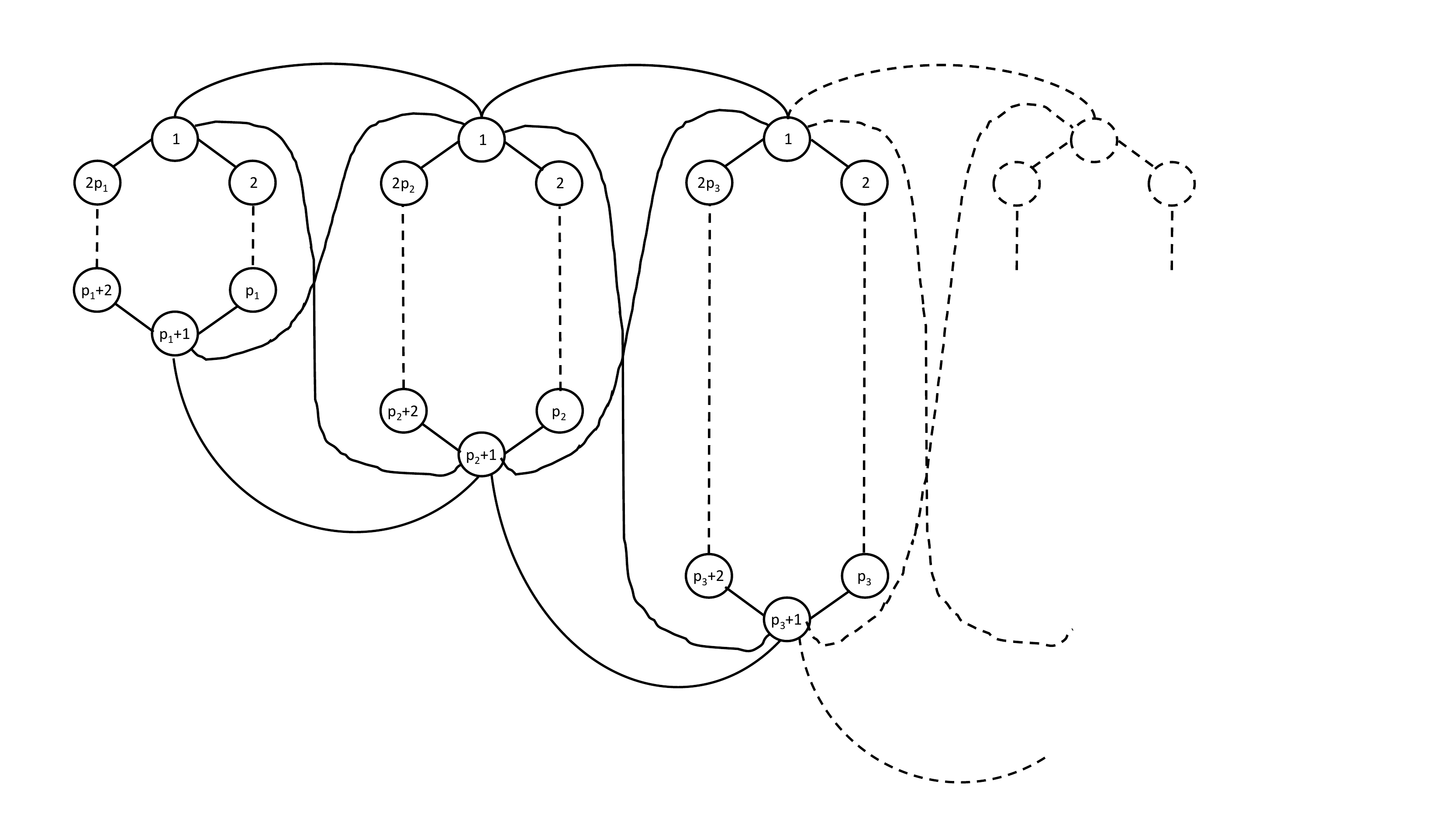} 
			\caption{Q2R network composed by the subnetworks $Q_1$,...,$Q_k$ as such of the proof of Proposition \ref{lem:r150-lp}, with prime sizes $p_1 < p_2 < p_3 <...< p_k$, respectively, and joined through its first and $p+1$ vertices.}
			\label{fig:q2r-superpol}
		\end{figure}
		
		Let $x$ be the configuration of length $2\cdot (p_1+p_2+\cdots +p_k)$, where all its states are in 0, except for the first vertices of each subnetwork $Q_i$, whose states are 1, $i\in\{1,...,k\}$. By the proof of Proposition \ref{lem:r150-lp} and Remark \ref{rem:q2r-2p-dynamic}-(i), each subnetwork $Q_i$ will produce a limit cycle of length $p_i$, $i\in\{1,...,k\}$. Therefore, the whole network $G$ will have a limit cycle of length $p_1\cdot p_2\cdots p_k$.

		To obtain a lower bound of the periods, we follow  the arguments developed in \cite{Kiwi:1994uj}.  Let $m$ a positive integer, and let $l = \pi(m)$  the number of primes not exceeding $m$. Let $G$ the graph obtained from $\pi(m)$ subnetworks of sizes $p_1, p_2 \dots, p_{\pi(m)}$, where $\{ p_1, p_2 \dots, p_{\pi(m)} \} $ the first $\pi(m)$ primes. We have then that 
		\begin{eqnarray}\label{eq3}
			V(G) \leq \sum_{i=1}^{\pi(m)} 2(p_i+1) \leq 2\pi(m)(m+1) 
		\end{eqnarray}
		and 
		\begin{eqnarray}\label{eq4}
			lcm(p_1, \dots, p_{\pi(m)}) = \prod_{i=1}^{\pi(m)} p_i =e^{\theta(m)} 
		\end{eqnarray}
		where $\theta(m) =  \sum_{i=1}^{\pi(m)} \log(pi)$. From the Prime Number Theorem \cite{hardy2008introduction}  we know that $\pi(m) = \Theta(m/\log(m))$, furthermore in \cite{hardy2008introduction}  is shown that $\theta(m) = \Theta(\pi(m)\log(m))$, which together with (\ref{eq3}) and (\ref{eq4})   imply that $$lcm(p_1, \dots, p_{\pi(m)})  \geq e^{\Omega(\sqrt{\vert V(G) \vert \log(\vert V(G) \vert) })}$$
		and then the length of the limit cycle of $G$ is not bounded by any polynomial in $\vert V(G) \vert$.

	\end{proof}

	\section{Computational Complexity of PRED}\label{sec:pred}

	
	We will give now a proof of {\bf P}-Hardness of the Q2R rule, reducing it to a restricted case of {\bf CVP}. 
	
	\begin{proposition}\label{lemma2} 
		For the Q2R rule, {\bf PRED} is {\bf P}-Hard.
	\end{proposition}
	\begin{proof}
		
		In order to show the result we will reduce {\bf S2MCVP} to \textbf{PRED}. Since {\bf S2MCVP} is \textbf{P}-complete, if the reduction uses only a logarithmic space, then {\bf PRED} would be \textbf{P}-hard.
		
		More precisely, we show that, given an arbitrary synchronous monotone alternating circuit of fanin and fanout $2$, there exists a Q2R network which simulates the evaluation of the circuit. Since the circuit is monotone, we only have to simulate the AND and OR gates. 
		
		
		
		The AND gate (see Figure \ref{ANDgate}) is simulated by an initially passive middle vertex with degree $4$; two of them will be the inputs, initially active and the other two the outputs, initially passive. By the Q2R rule, this vertex will become active only if the two input neighbors become active. The OR gate is obtained as a composition of the NOT gate and the AND gate as it is shown in the next section. Now, consider an instance of \textbf{S2MCVP} given by some circuit $C,$ an assignation $x \in \{0,1\}^{n},$ and some output $o$. We construct a Q2R network $\mathcal{Q}_{C} = (G=(V=(B,W),E), F)$ which has for each gate $v$ of $C$, the correspondent gadget. Let us call $(i_{1},..., i_{n})$ to the inputs of $C$. We identify one gadget in $\mathcal{Q}_{C}$ with each of those inputs. Observe that each gadget has its own input nodes as it is shown in the next section. Thus, $(i_{1},..., i_{n})$ can be identified with the inputs of the gadgets representing the first layer of the circuit. We define an initial condition $\overline{x} \in \{0,1\}^{V}$ given by $\overline{x}_{i_{k}} = x_{k}$ for $1 \leq k \leq n.$ and $\overline{x}_{v} = 0$ for any other node $v.$ Since each gate is monotone, each gadget is constant in the number of size of $C$ (and also they simulate each gate in constant time) and the circuit is synchronous, $\mathcal{Q}_{C}$ simulates $C$ in time $t = n^{\mathcal{O}(1)}.$ Finally, as the size of each gadget is constant, the latter construction can be done in space $\mathcal{O}(\log n)$. The proposition holds.
	\end{proof}
	
	\section{Gadgets to show  the \textbf{P}-hardness of \textbf{PRED}}\label{sec:gadgets} 
	%
	
	\noindent {\bf Wire:} the 1s (black nodes) go down following the wire (starting from top left, the first 4 iterations of Figure \ref{wire}).
	\begin{figure}[h] 
		\begin{center}    
			\includegraphics[scale=0.5]{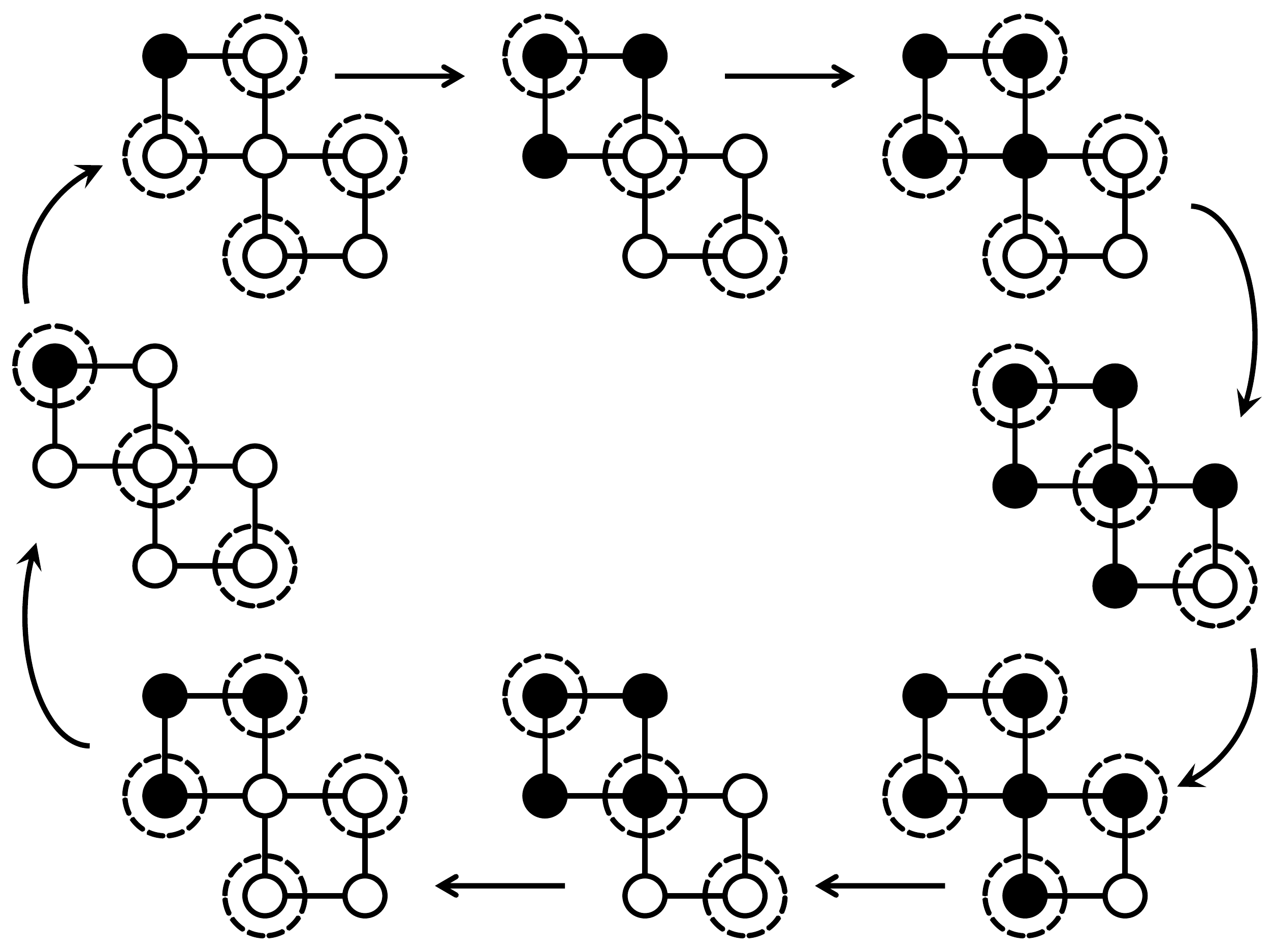}
		\end{center}
		\caption{Evolution of a 7 node wire where blacks ones indicate states at 1 and the whites ones at -1. The dashed circles indicate the nodes that will be updated at each time step.}
		\label{wire}
	\end{figure}
	
	\noindent {\bf Remark:} if the wire is finite, say the case of Figure \ref{wire} with two squares, when the last ``down'' is reached, the -1s works as the 1s ``unweaving'' the weaving (starting from top left, see the whole cycle of period 8 in Figure \ref{wire}).  
	
	\subsection{The AND gate}\label{sec:ANDgate}
	Consider the wire of Figure \ref{ANDgate} where $x,y\in\{-1,1\}$ are the inputs, $z$ is the output (initially at -1) and the remaining ten nodes with -1 values (white nodes). 
	\begin{figure}[h] 
		\begin{center}    
			\includegraphics[scale=0.3]{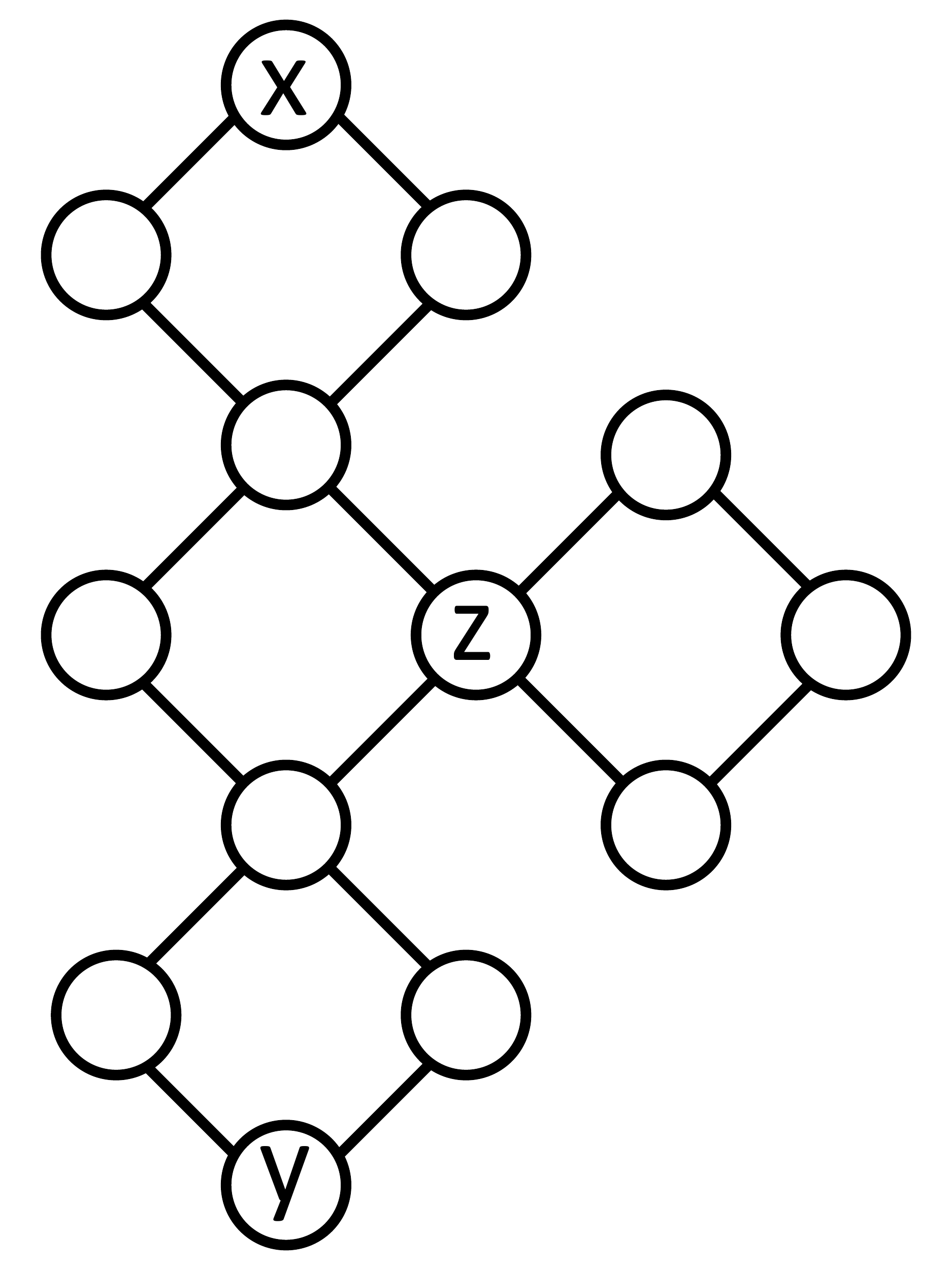}
		\end{center}
		\caption{Structure of the AND wire described in section \ref{sec:ANDgate} and which simulates an AND gate.}
		\label{ANDgate}
	\end{figure}
	
	So, it is easy to check that such a wire simulates an AND gate. In fact:
	\begin{itemize}
		\item If $x=y=1$ then in 3 time steps $z=1$ (see Figure \ref{ANDwire11}). 
		\item If $x=1$ and $y=-1$ then $z=-1$ always (see Figure \ref{ANDwire10}). Case $x=-1$ and $y=1$ is analogous.
		\item The case $x=y=-1$ is straightforward because the initial configuration would be the all -1s fixed point.
	\end{itemize}
	
	\begin{figure}[h] 
		\begin{center}    
			\includegraphics[scale=0.5]{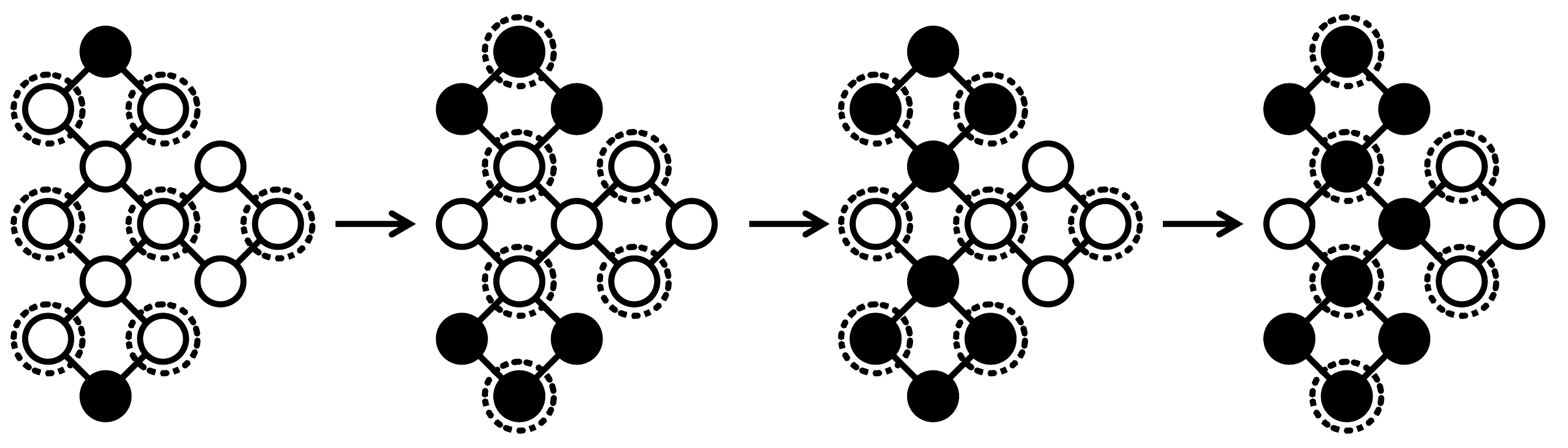} 
		\end{center}
		\caption{(Partial) evolution of the AND wire of Figure \ref{ANDgate} with input $x=y=1$ (black). The dashed circles indicate the nodes that will be updated at each time step.}
		\label{ANDwire11}
	\end{figure}
	
	\begin{figure}[h] 
		\begin{center}    
			\includegraphics[scale=0.5]{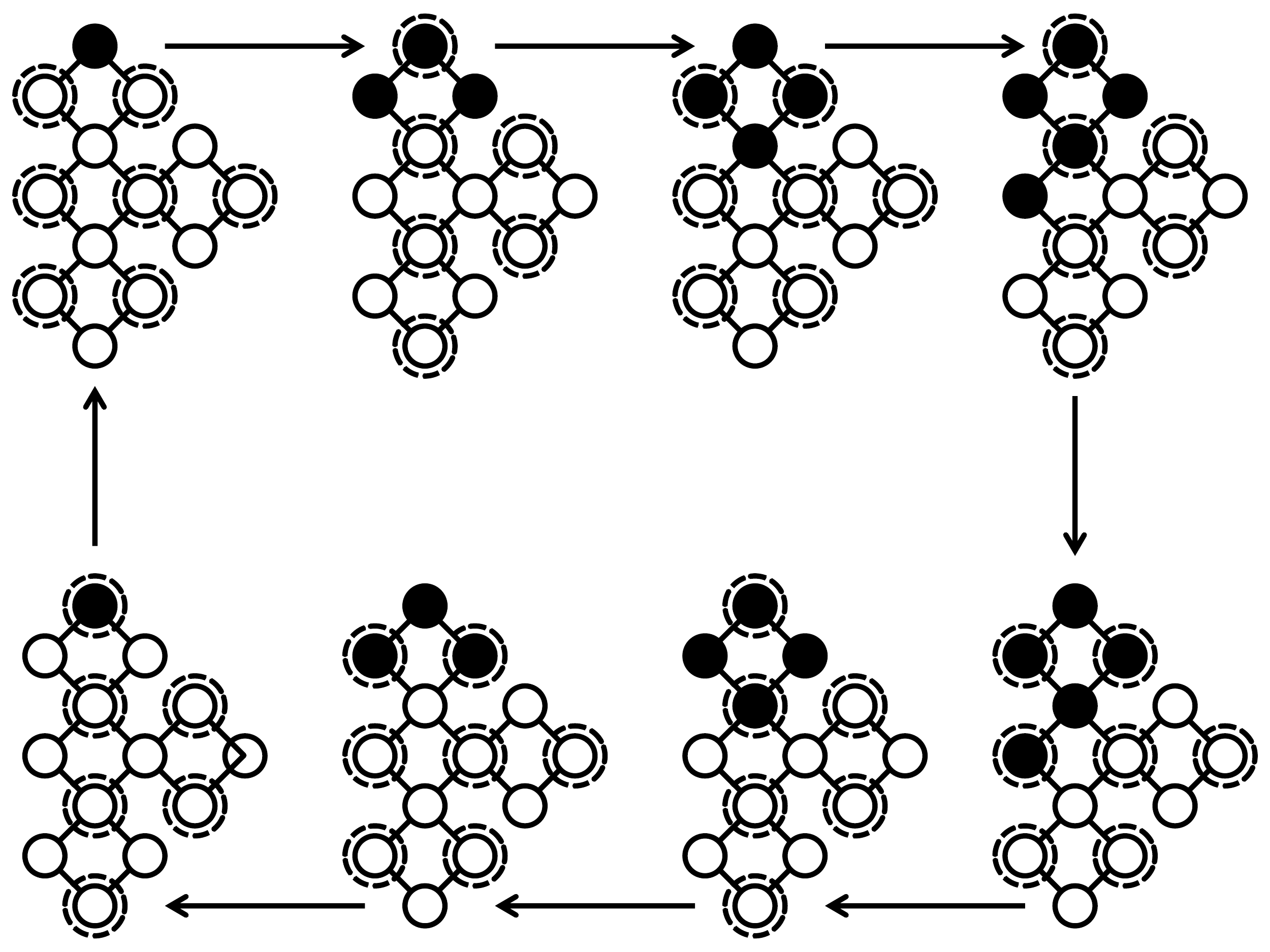} 
		\end{center}
		\caption{Evolution of the AND wire of Figure \ref{ANDgate} with input $x=1$ (black) and $y=-1$ (white). The dashed circles indicate the nodes that will be updated at each time step.}
		\label{ANDwire10}
	\end{figure}
	
	\subsection{The XOR gate}\label{sec:XORgate}
	Consider the wire of Figure \ref{XORgate} where $x,y\in\{-1,1\}$ are the inputs and $z$ is the output (initially at -1). The idea is to fix as stable sites two 1s. The inputs $x$, $y$ comes from two wires as above. 
	\begin{figure}[h] 
		\begin{center}   
			\includegraphics[scale=0.5]{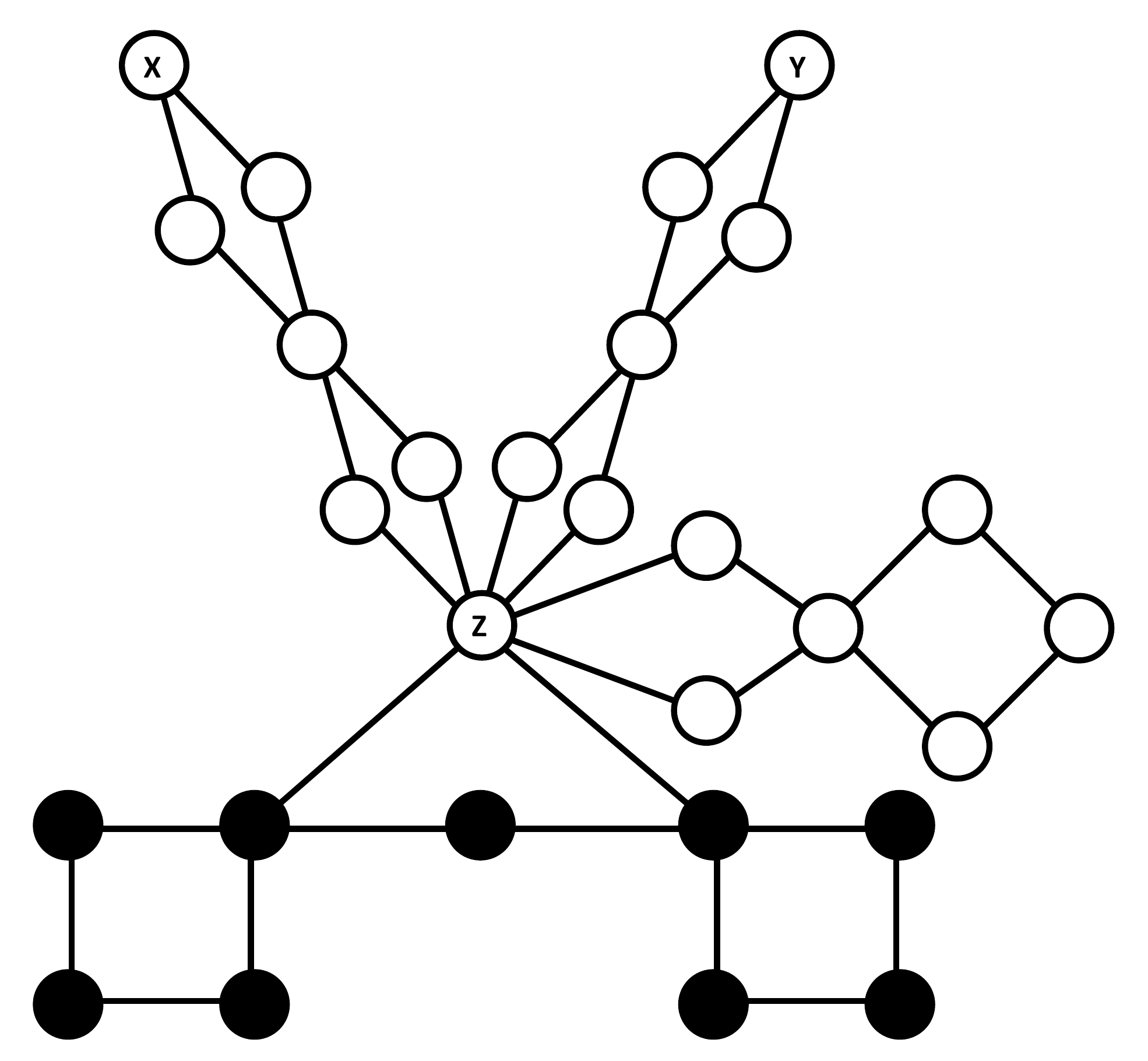} 
		\end{center}
		\caption{Structure of the XOR wire described in section \ref{sec:XORgate} and which simulates an XOR gate.}
		\label{XORgate}
	\end{figure}
	
	Clearly, $x+y=1 \Leftrightarrow [(x=1)\wedge (y=-1)] \vee [(x=-1)\wedge (y=1)]$. Moreover:
	\begin{itemize}
		\item If $x=1$ and $y=-1$, then $z=1$ in 3 time steps.
		\item If $x+y=2$, i.e., $x=y=1$, then the output $z$ remains at -1. 
	\end{itemize}
	
	\subsection{The NOT gate}\label{sec:NOT}
	Here we consider the fact $\neg x =1\oplus x$, so, it's enough to adapt the previous gadget, as in Figure \ref{NEGgate} (top left) where $x$ is the input and the output it's by the wire on the right.
	\begin{figure}[h] 
		\begin{center}   
			\includegraphics[scale=0.4]{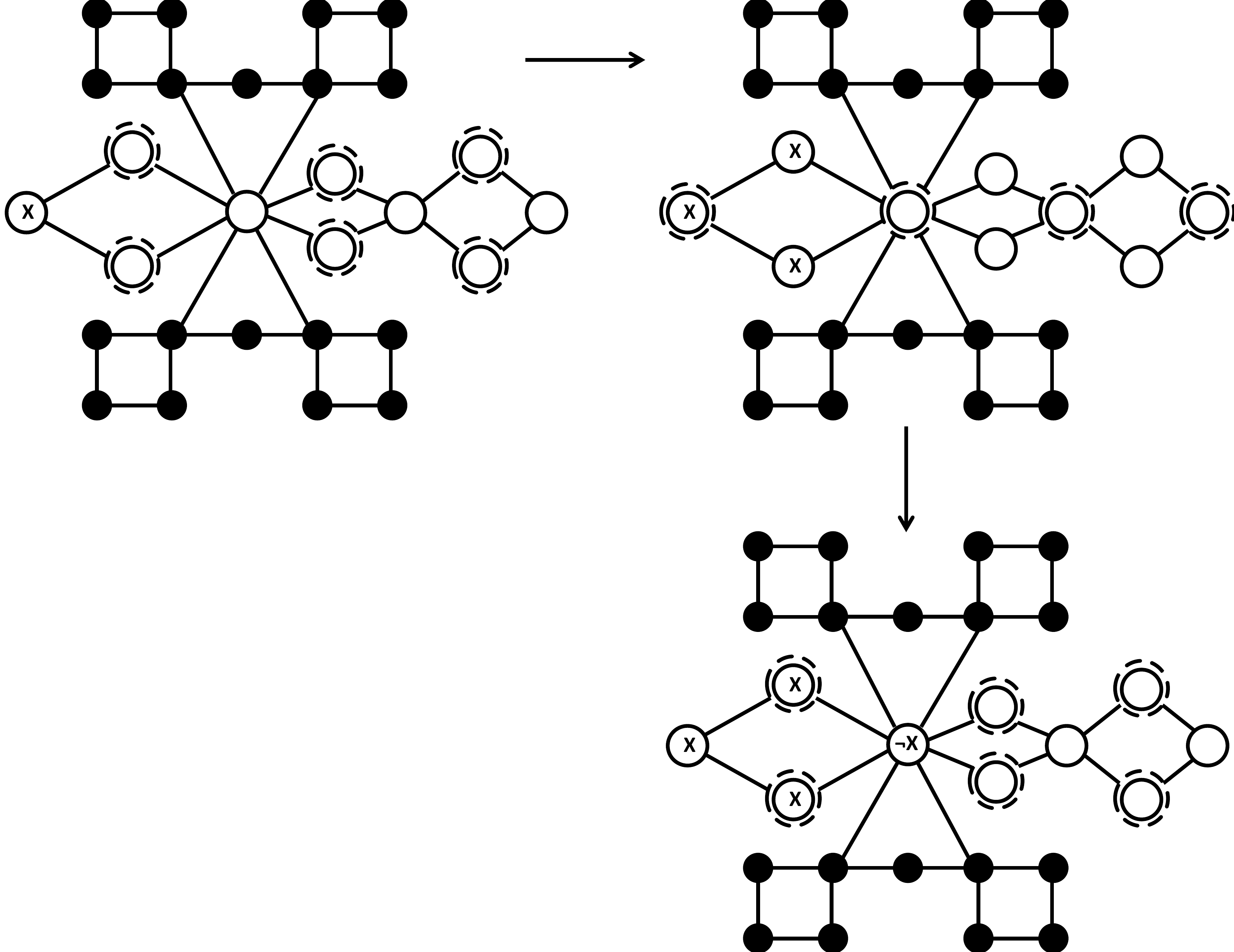} 
		\end{center}
		\caption{(Partial) evolution of the NEGATION wire described in section \ref{sec:NOT} with input $x\in\{-1,1\}$. The dashed circles indicate the nodes that will be updated at each time step.}
		\label{NEGgate}
	\end{figure}
	In Figure \ref{NEGgate} we schematize the dynamics that shows how an input $x\in\{-1,1\}$ is transformed into the output $\neg x$.
	
	\subsection{The OR gate}
	One may use the Morgan law: $x\vee y =\neg (\neg x \wedge \neg y)$ and to use the previous AND and NEGATION gates.
	
	\subsection{Cross-Over}
	Consider the structure of Figure \ref{crossover}., from that, by using previous  XOR gate, one constructs the cross-over. 
	\begin{figure}[h] 
		\begin{center}   
			\includegraphics[scale=0.5]{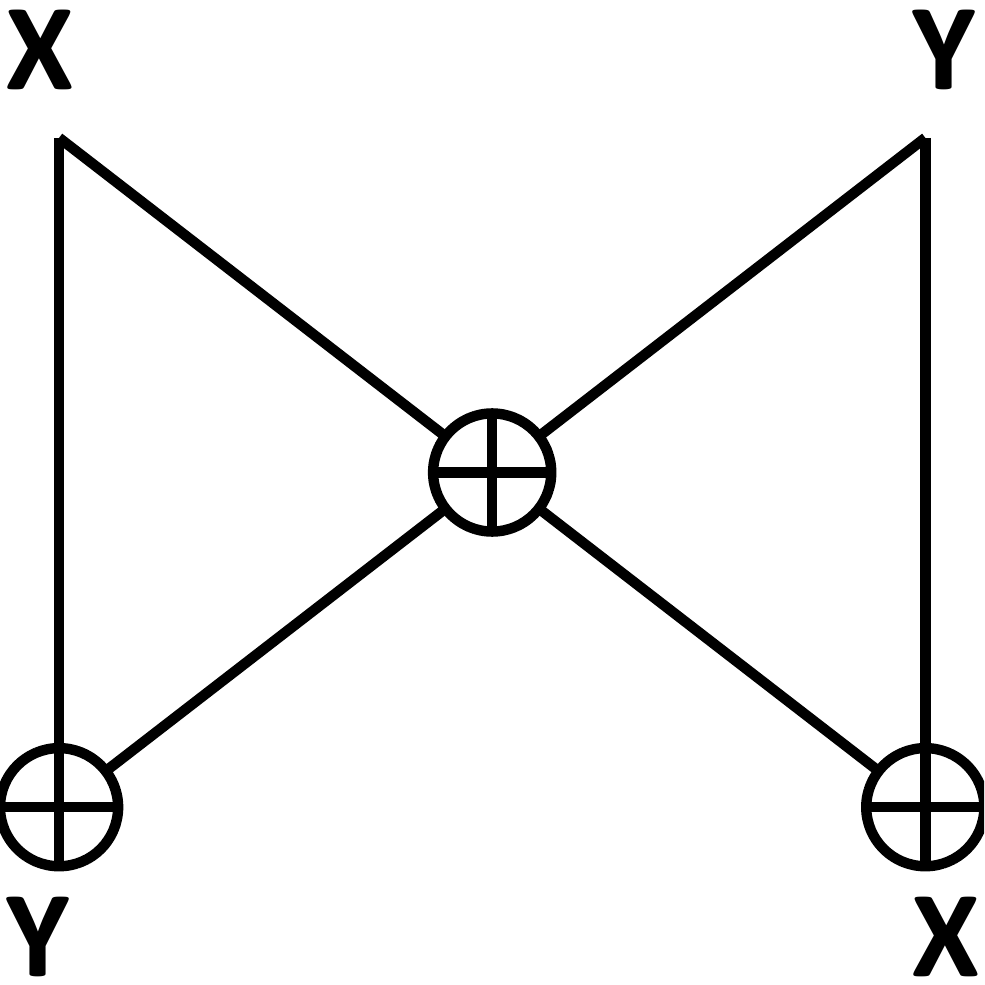} 
		\end{center}
		\caption{Structure of the cross-over.}
		\label{crossover}
	\end{figure}
	
	%
	%
	
	\section{Simulating Q2R with parallel update schedule}\label{sec:ps}
	Let $G=(V,E)$ be a Q2R network with both $n=|V(G)|$ and the degree of each node being even. In Figure \ref{sp-simulator} we present the general structure of what we will call the {\it Parallel Simulator (PS)}, whose function will be to simulate, with parallel update, the dynamics of the previous Q2R network updated by the block update $s=(A)(B)$, where $A$ and $B$ are two non-empty subsets of nodes such that $V=A\cup B$.
	\begin{figure}[h] 
		\centering
		\includegraphics[scale=0.2]{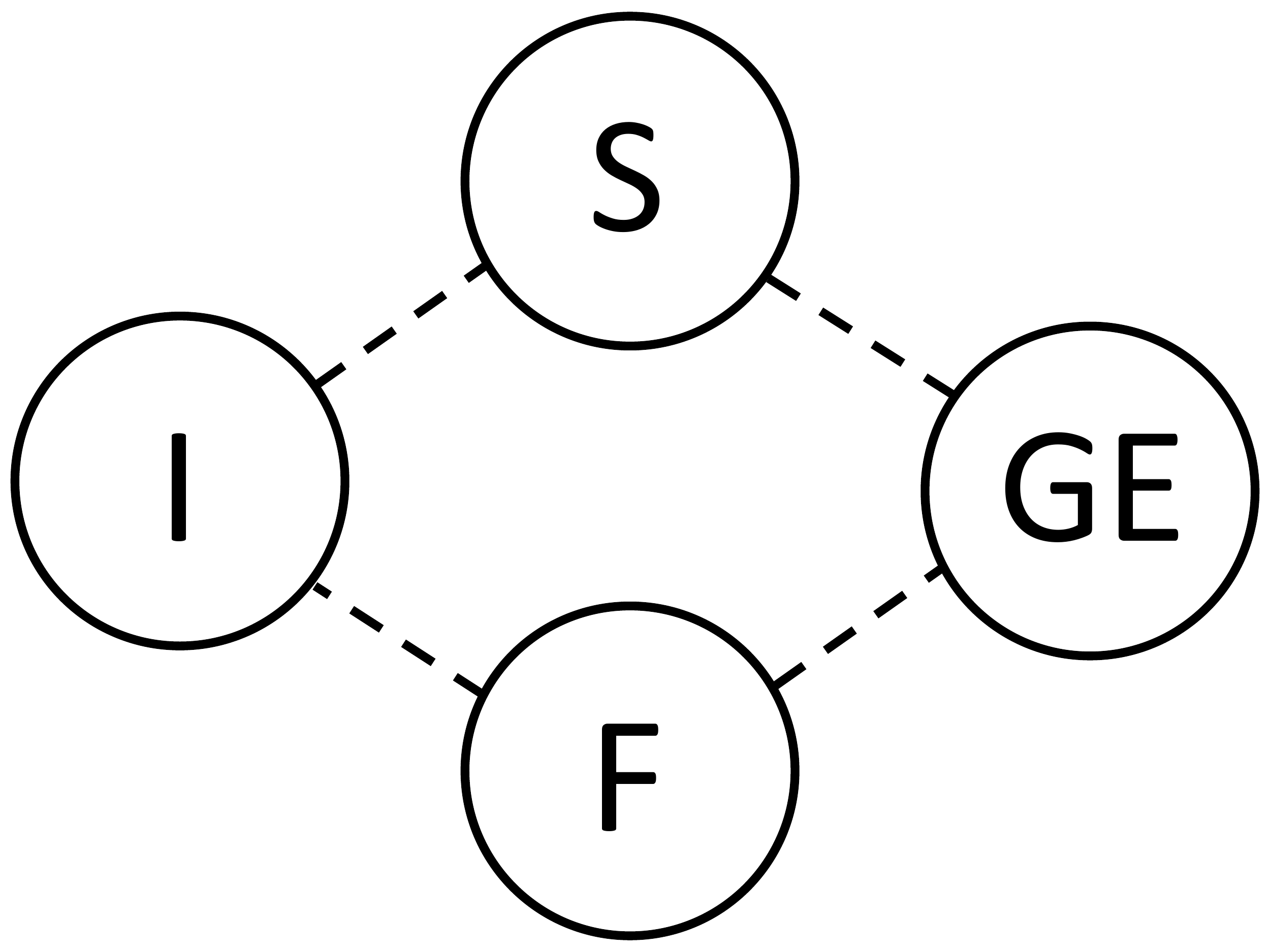}
		\caption{General structure of the Parallel Simulator (PS)}
		\label{sp-simulator}
	\end{figure}
	
	Basically, the PS is another Q2R network that also keeps the parity in both, its number of nodes and in its degrees. Furthermore, it possesses four components, denoted by $I$, $S$, $F$ and $GE$, that are connected in a certain way (see details in section \ref{sec_PS-connections}) and whose names and characteristics are described below:
	\begin{enumerate}
		\item $I$, the {\it Input component}, it contains the graph to be simulated in parallel.
		\item $S$, the {\it Switching component}, is a graph that, for a given initial configuration, always generates the same limit cycle, regardless of the connections it may have with the other PS components.
		\item $F$, the {\it Fixed component}, is a graph that, for a given initial configuration, never changes, regardless of the connections it may have with the other PS components.
		\item $GE$, the {\it Gear component}, is a graph that, for a given initial configuration, makes the PS works out.
	\end{enumerate}
	
	Next, we specify the above components as well as their connections and necessary notations to demonstrate that PS can simulate in parallel the dynamics of the Q2R network $G$ updated by the block update $s=(A)(B)$ of above.
	
	\subsection{Specifying the PS components}\label{sec_PS-components}
	\begin{itemize}
		\item The input component will be $G$, i.e., $I=G$ and we denote by $\alpha$ the half of the maximum degree of $G$, i.e., $\alpha=\frac{\max\{d(v):\,v\in V(G)\}}{2}$.
		\item The switching component has $\alpha$ identical graphs, denoted by $S^i$, $i\in\{1,...,\alpha\}$, of four nodes each. For $i\in\{1,...,\alpha\}$, we denote by $s_{1}^{i}$, $s_{2}^{i}$, $s_{3}^{i}$ and $s_{4}^{i}$ such a nodes. The graph $S^i$ and its initial configuration are showed in Figure \ref{fig_Scomp} a) and b), respectively. Thus, the limit cycle we want to keep will be the one of length two that alternates between the initial and final configurations showed in Figure \ref{fig_Scomp} b) and c), and that can be obtained by iterating, in an isolating way, any graph $S^i$.
		\begin{figure}[htbp!]
			\begin{center}
				\begin{tabular}{lll}
					a) \includegraphics[scale=.15]{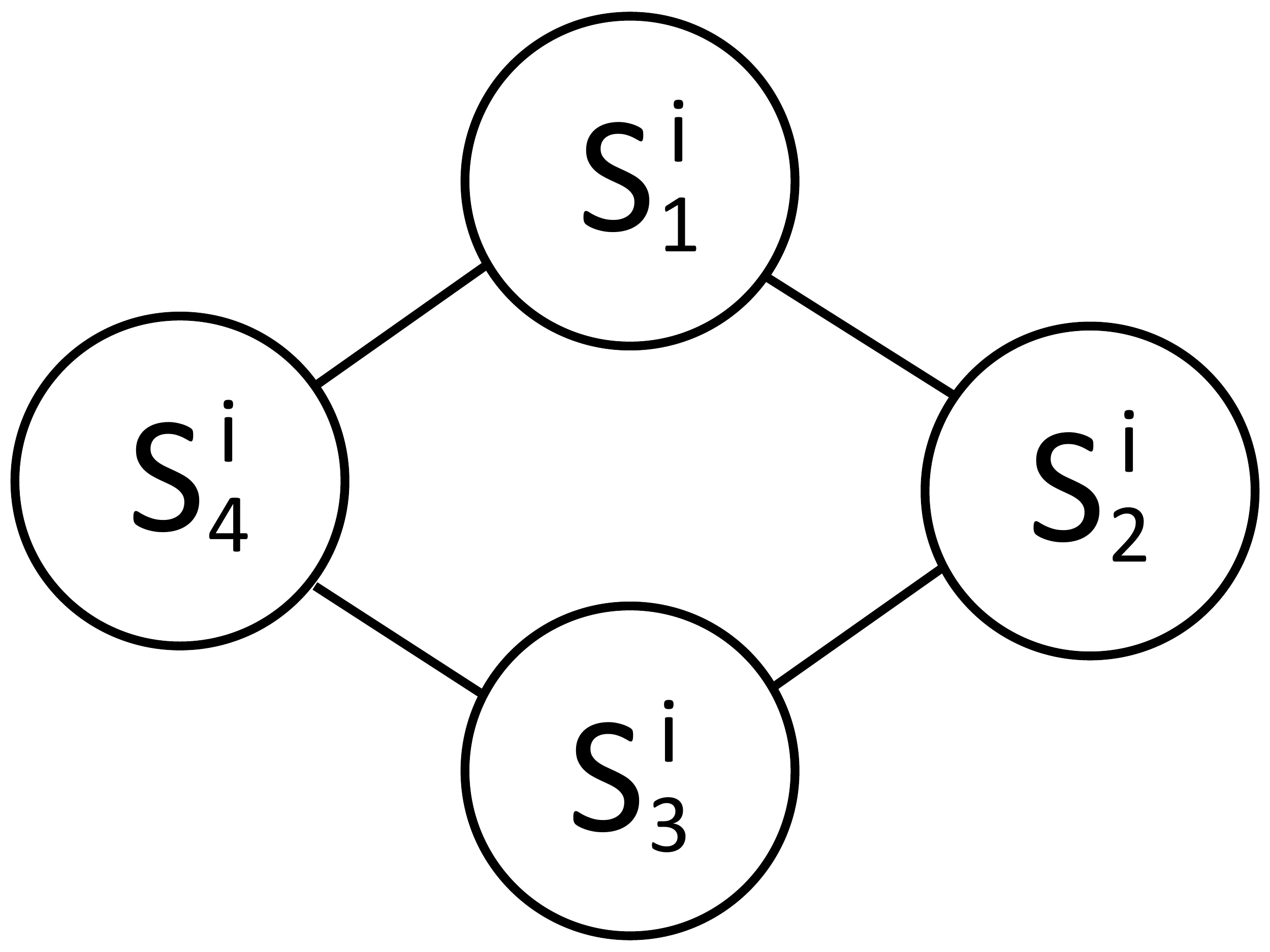} & b) \includegraphics[scale=.15]{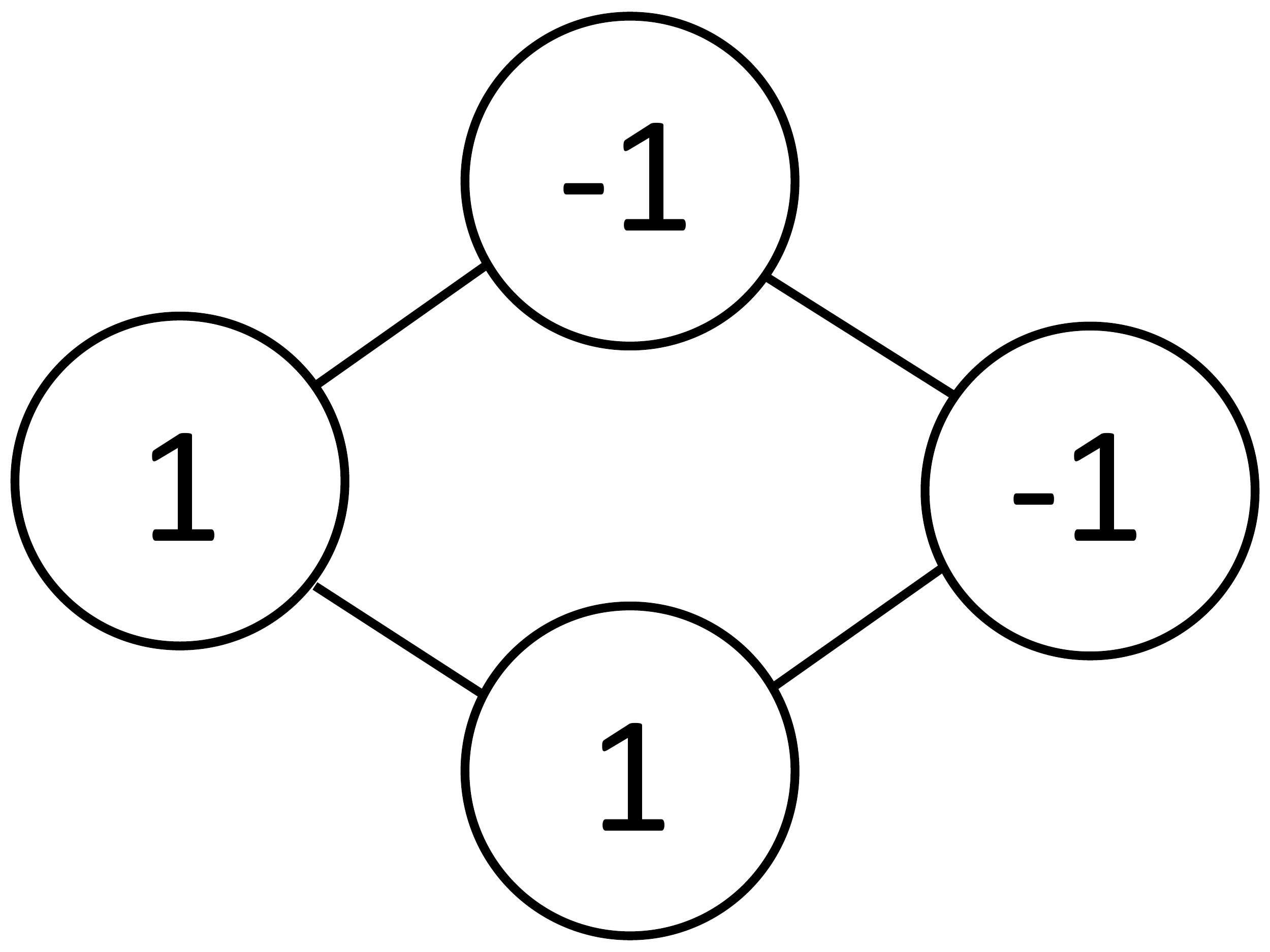} & c) \includegraphics[scale=.15]{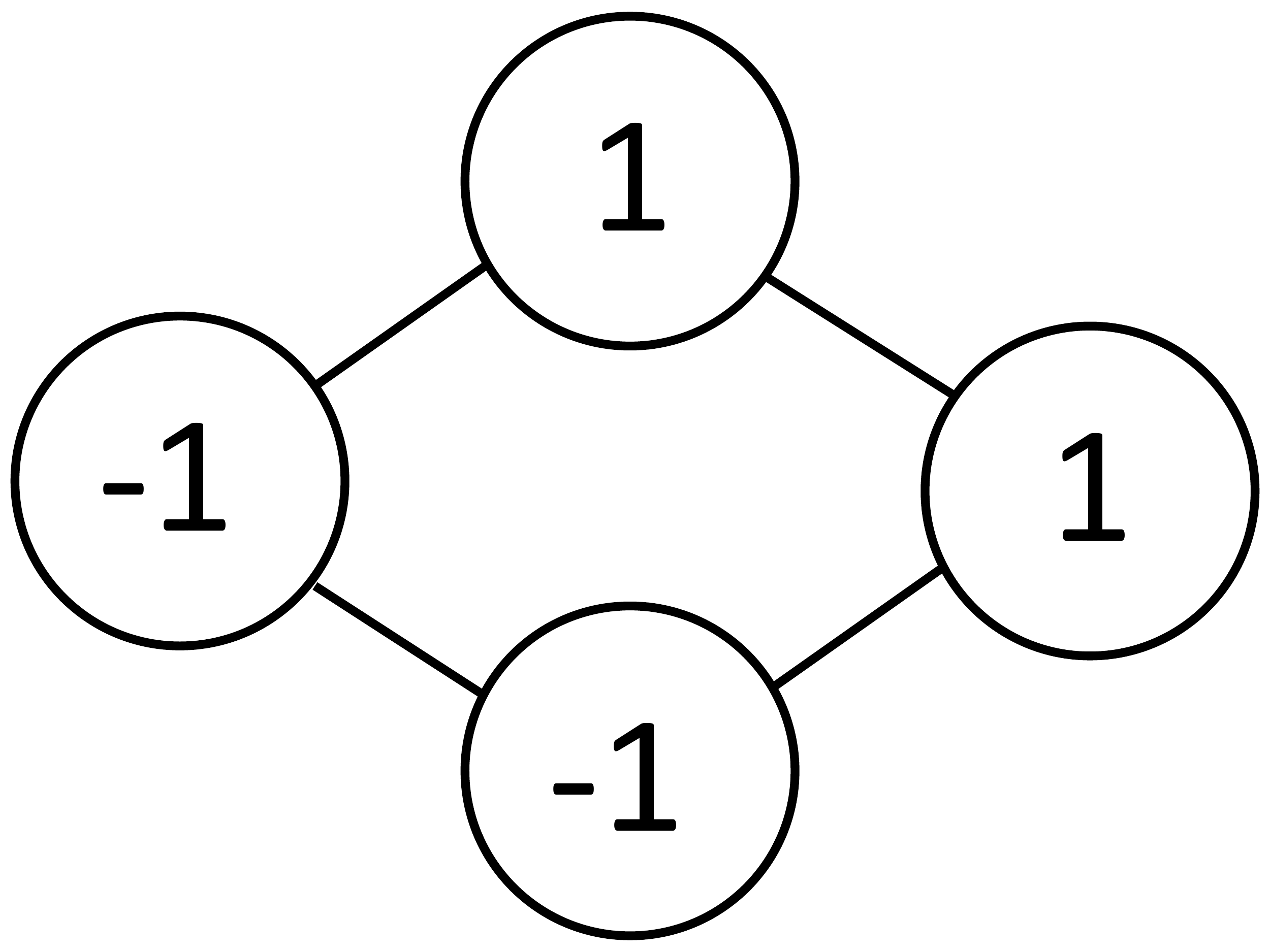}
				\end{tabular}
				\caption{a) A graph $S^i$ of the switching component. b) Initial configuration in $S^i$. c) Final configuration in $S^i$.} \label{fig_Scomp}
			\end{center}
		\end{figure}
		\item The fixed component has $\alpha$ identical graphs, denoted by $F^i$, $i\in\{1,...,\alpha\}$, of four nodes each. For $i\in\{1,...,\alpha\}$, we denote by $f_{1}^{i}$, $f_{2}^{i}$, $f_{3}^{i}$ and $f_{4}^{i}$ such a nodes. The graph $F^i$ and its initial configuration that we want to keep without changes are showed in Figure \ref{fig_Fcomp} a) and b). Observe that the initial configuration is a fixed point of any $F^i$ dynamics when it is iterated in an isolating way.
		\begin{figure}[htbp!]
			\begin{center}
				\begin{tabular}{ll}
					a) \includegraphics[scale=.2]{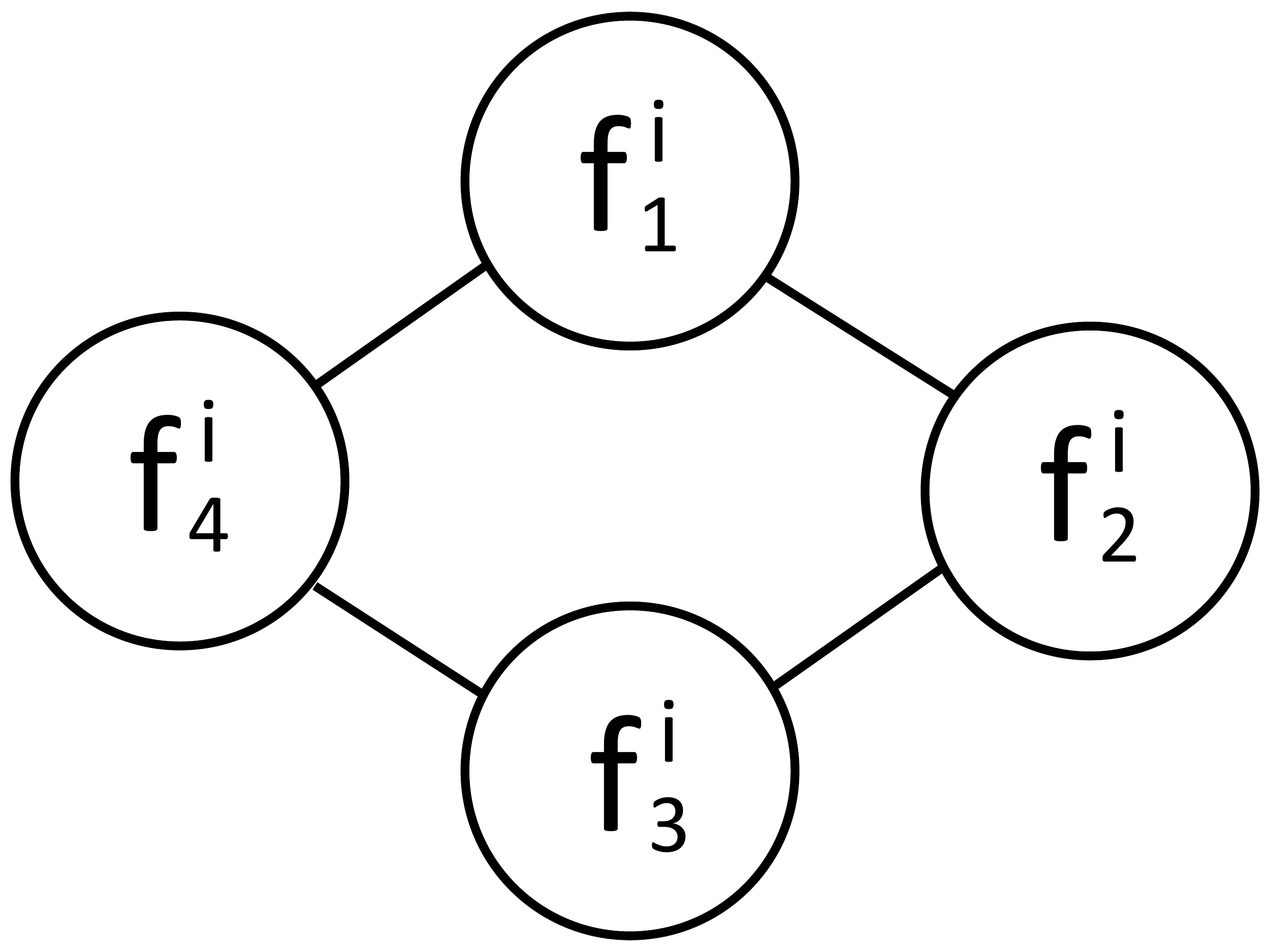} & b) \includegraphics[scale=.2]{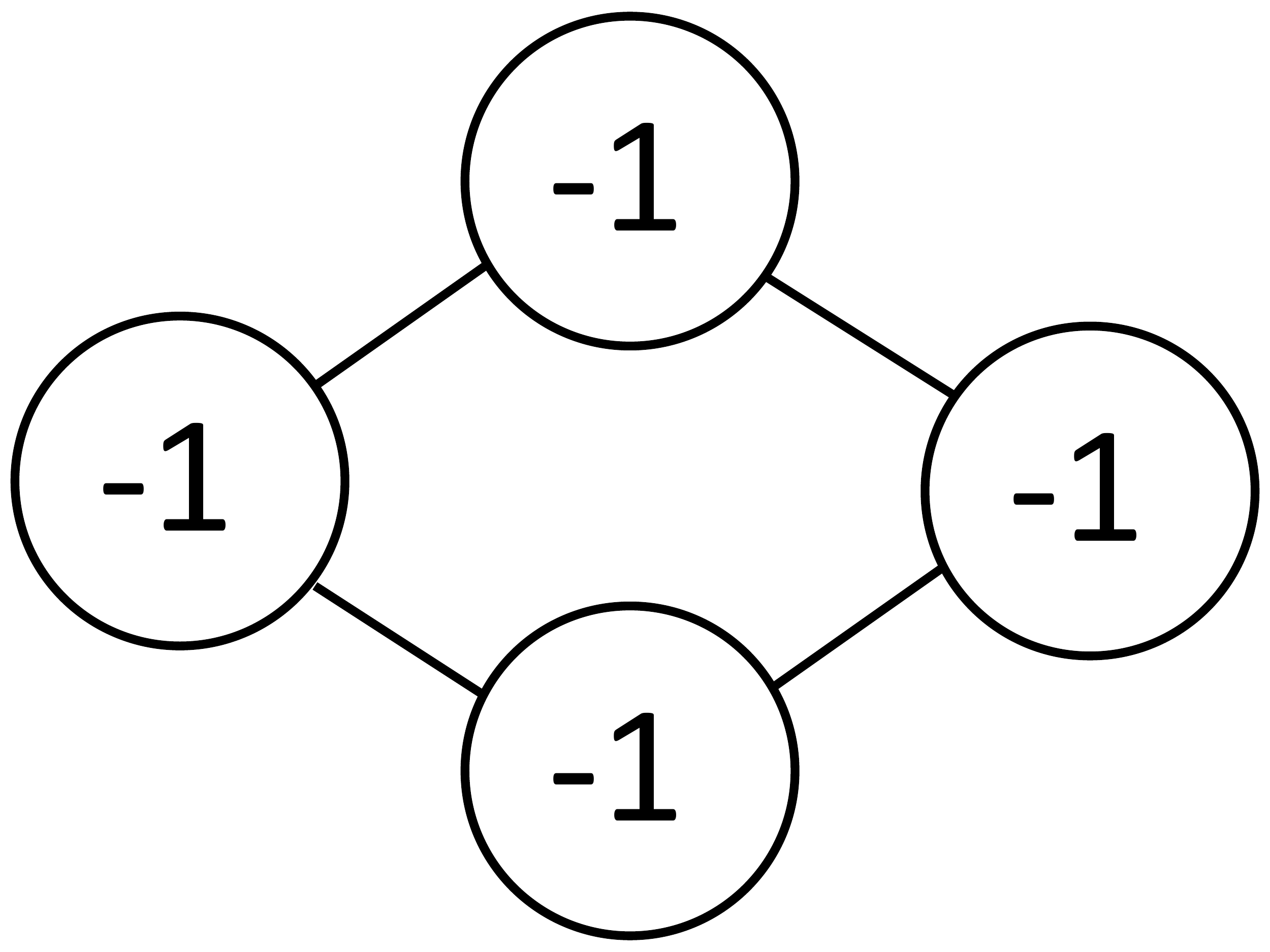} 
				\end{tabular}
				\caption{a) A graph $F^i$ of the fixed component. b) Initial configuration (fixed) of $F^i$.} \label{fig_Fcomp}
			\end{center}
		\end{figure}
		\item The gear component will be the same network $G$ but with the negation of all the states of the initial configuration of $G$, i.e., if we denote by $\overline{1}$,...,$\overline{n}$ the nodes of the gear component, then $x_{\overline{i}}=\neg x_i$, for all $i\in\{1,...,n\}$.
	\end{itemize}
	
	\subsection{Specifying the PS connections}\label{sec_PS-connections}
	Let's consider the Q2R network $G$ updated by the block update $s=(A)(B)$, where $A$ and $B$ are two non-empty subsets of nodes such that $V=A\cup B$. Moreover, we denote $\overline{A}=\{\overline{u}:\,u\in A\}$ and $\overline{B}=\{\overline{v}:\,v\in B\}$. Besides, given the nodes $u\in A$ and $v\in B$, we denote by $k$ and $j$ the half of its degrees in $G$, respectively, i.e., $k=\frac{d(u)}{2}$ and $j=\frac{d(v)}{2}$. Suppose w.l.o.g. that $1\leq j\leq k=\alpha$. Then, we consider the following PS connections (see Figure \ref{fig:PS-connections}):
	\begin{itemize}
		\item[(a)] The edges $\{u,s^i_3\}$, $\{u,s^i_4\}$ and $\{\overline{u},s^i_3\}$, $\{\overline{u},s^i_4\}$, for $i\in\{1,...,k=\alpha\}$.
		\item[(b)] The edges $\{v,s^i_1\}$, $\{v,s^i_2\}$ and $\{\overline{v},s^i_1\}$, $\{\overline{v},s^i_2\}$, for $i\in\{1,...,j\}$.
		\item[(c)] The edges $\{u,f^i_1\}$, $\{u,f^i_4\}$ and $\{\overline{u},f^i_1\}$, $\{\overline{u},f^i_4\}$, for $i\in\{1,...,k=\alpha\}$.
		\item[(d)] The edges $\{v,f^i_2\}$, $\{v,f^i_3\}$ and $\{\overline{v},f^i_2\}$, $\{\overline{v},f^i_3\}$, for $i\in\{1,...,j\}$.
	\end{itemize}
	\begin{figure}[h] 
		\centering
		\includegraphics[scale=0.5]{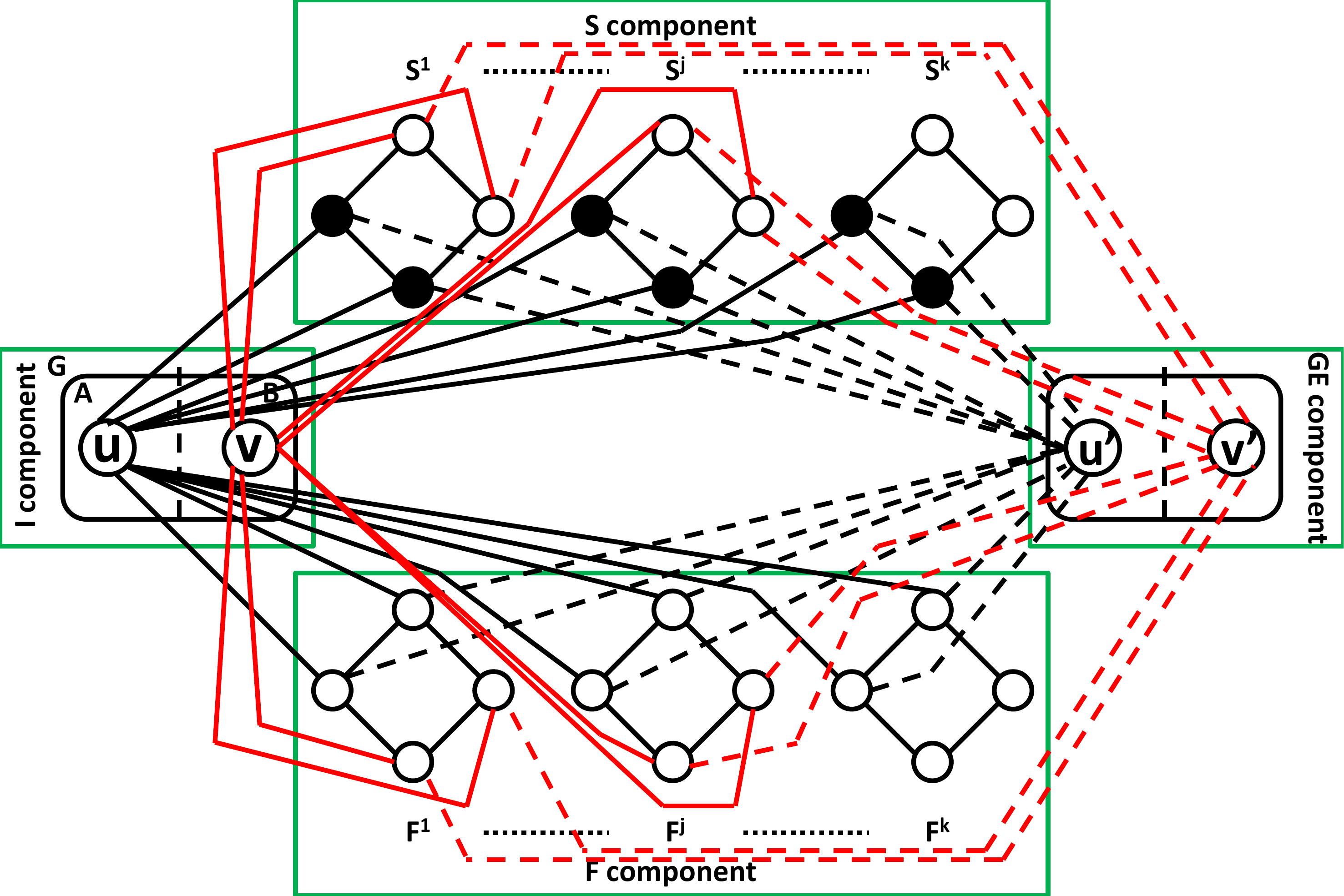} 
		\caption{The specific connections of the PS described in sections \ref{sec_PS-components} and \ref{sec_PS-connections}.}
		\label{fig:PS-connections}
	\end{figure}
	
	\begin{theorem}\label{teo-PS}
		One time step of the Q2R network $G$ updated by the block sequential $s=(A)(B)$ correspond to two time steps of the Q2R network $G$ updated by the PS in parallel.
	\end{theorem}
	\begin{proof}
		Let the PS with its components and connections specified in sections \ref{sec_PS-components} and \ref{sec_PS-connections} at $t=0$. In this time step, we have that:
		\begin{enumerate}
			\item The PS connections defined in (a) and (b) guarantee that the sum of the neighbors of $s^i_j$ is 0, $\forall i\in\{1,...,\alpha\}$, $\forall j\in\{1,...,4\}$. 
			\item The PS connections defined in (c) and (d) guarantee that the sum of the neighbors of $f^i_j$ is different to 0 (more precise, -2), $\forall i\in\{1,...,\alpha\}$, $\forall j\in\{1,...,4\}$.
			\item The edges $\{u,s^i_3\}$, $\{u,s^i_4\}$, $\{u,f^i_1\}$ and $\{u,f^i_4\}$, $i\in\{1,...,\frac{d(u)}{2}\}$, $\frac{d(u)}{2}\leq\alpha$, contribute with 0 to the sum of the neighbors of $u$, for all $u\in A$, i.e., the update of $u$ only depends on the values of its neighbors in $G$. The same occurs if we change the roles of $u$ and $A$ by $\overline{u}$ and $\overline{A}$, respectively.
			\item The edges $\{v,s^i_1\}$, $\{v,s^i_2\}$, $\{v,f^i_2\}$ and $\{v,f^i_3\}$, $i\in\{1,...,\frac{d(v)}{2}\}$, $\frac{d(v)}{2}\leq\alpha$, contribute with $-4\cdot\left(\frac{d(v)}{2}\right)$ to the sum of the neighbors of $v$, for all $v\in B$. Thus, such a sum will be upper bounded by $-d(v)$, i.e., this sum will be different to 0. The same occurs if we change the roles of $v$ and $B$ by $\overline{v}$ and $\overline{B}$, respectively.
		\end{enumerate}
		
		Hence, at $t=1$: 
		\begin{enumerate}
			\setcounter{enumi}{4}
			\item All the $\alpha$ graphs of the switching component will change to the final configuration of Figure \ref{fig_Scomp}c) (because of point 1), having again what was said in point 1.
			\item All the $\alpha$ graphs of the fixed component will not change, i.e., they will continue with the same values of Figure \ref{fig_Fcomp}b) (because of point 2), having again what was said in point 2.
			\item Only the nodes of $A$ and $\overline{A}$ will have been updated (because of point 3) while those of $B$ and $\overline{B}$ remain unchanged (because of point 4).
			\item The edges $\{u,s^i_3\}$, $\{u,s^i_4\}$, $\{u,f^i_1\}$ and $\{u,f^i_4\}$, $i\in\{1,...,\frac{d(u)}{2}\}$, $\frac{d(u)}{2}\leq\alpha$, contribute with $-4\cdot\left(\frac{d(u)}{2}\right)$ to the sum of the neighbors of $u$, for all $u\in A$. Thus, such a sum will be upper bounded by $-d(u)$, i.e., this sum will be different to 0. The same occurs if we change the roles of $u$ and $A$ by $\overline{u}$ and $\overline{A}$, respectively.
			
			\item The edges $\{v,s^i_1\}$, $\{v,s^i_2\}$, $\{v,f^i_2\}$ and $\{v,f^i_3\}$, $i\in\{1,...,\frac{d(v)}{2}\}$, $\frac{d(v)}{2}\leq\alpha$, contribute with 0 to the sum of the neighbors of $v$, for all $v\in B$, i.e., the update of $v$ only depends on the values of its neighbors in $G$. The same occurs if we change the roles of $v$ and $B$ by $\overline{v}$ and $\overline{B}$, respectively. 
		\end{enumerate} 
		Therefore, at $t=2$, the PS will have exactly the first time step of the Q2R network $G$ updated by the block sequential $s=(A)(B)$.
	\end{proof}
	
	\section{Discussion}
	

	In this article we have studied some dynamic and complexity properties of the generalized Q2R automaton. It has been shown that it admits non-polynomial cycles and, by considering the parallel update, it may simulate the classical Q2R model (which is based in a block-sequential update scheme, following a bipartite partition of vertices). It is important to point out that the general structure of Parallel Simulator (PS) can be extended to simulate the dynamics of any other Q2R network under an arbitrary block-sequential update scheme $s'=(B_1)\cdots(B_k)$, $1<k\leq n.$ These latter task can be accomplished by choosing graphs in the switching component that generate appropriate cycles of length $k$ so that, at the time step $t$, the input component of the PS only updates the nodes of $B_t$ for $1\leq t\leq k.$ Thus, one time step of the Q2R network updated by the block sequential $s'=(B_1)\cdots(B_k)$, $1<k\leq n$, correspond to $k$ time steps of the Q2R network $G$ updated by the PS. 
	
	Furthermore, from a computational complexity standpoint, we shown that the problem of predicting whether the state of a cell changes in $t$ steps is \textbf{P}-Hard. It should be noted that such complexity is associated with the degree of vertices in the bipartite graph. In fact, the one dimensional (finite) Q2R cellular automaton corresponds to the elementary cellular automaton given by the rule $150$ (which is defined by a XOR local function depending of the left, center and right cells) defined over a ring. Since the latter global rule is a linear function (when considered as a function over $(\{0,1\},+,\cdot)$, where $+$ is the XOR function and $\cdot$ is the AND function), it is possible to represent it as a Boolean matrix (which coincides with the adjacency matrix of the graph). Thus, we can compute efficiently $t$ time steps of the dynamics (and in particular, the state of a cell) by simply computing the $t$-th power of this matrix. Nevertheless,  the computational complexity of this prediction problem in the two dimensional grid remains still open. In fact, even when our wire,  AND and OR gadgets can be implemented on such a grid, it is not possible to directly implement the cross over gadget since it needs the XOR gadget, which in our construction, requires $8$ neighbors. However, we observe that, if the condition on the even degree of each vertex is relaxed, which is equivalent to consider freezing sites (remaining fixed during the Q2R dynamics because odd degree implies no ties) then, smaller gadgets can be implemented. We believe that in this case, a XOR gadget can be embedded in a subgraph of the two dimensional grid. Still, even if we succeed on implementing that latter gadget, it is a very different problem far from the essence of the classical Q2R.

	
	
	\bibliographystyle{elsarticle-num}
	\bibliography{bibliov2}
	
	
	
	
	
	

\end{document}